\documentclass[fleqn,11pt,twoside]{article}

\usepackage{amsthm,amsthm,amssymb, color, xcolor,epsfig, graphics, subfigure}

\usepackage{amsmath, graphicx, latexsym, lscape}

\makeatletter
\newcommand{\copyrightnote}[2]{{\renewcommand{\thefootnote}{}
 \footnotetext{\small\it
\begin{flushleft}
 \copyright \ #1   #2  
\end{flushleft}}}}

\newcommand{\Name}[1]{\begin{flushleft}
                       \LARGE \bf #1
                       \end{flushleft}\vspace{-3mm}}

\newcommand{\Author}[1]{\begin{flushleft}
                       \it #1 \end{flushleft}}

\newcommand{\Address}[1]{\begin{flushleft}
                       \it #1 \end{flushleft}}

\newcommand{\Date}[1]{\begin{flushleft}
                      \small  \it #1 \end{flushleft}}

\newcommand{\evenhead}{Author \ name}
\newcommand{\oddhead}{Article \ name}

\renewcommand{\@evenhead}{
\hspace*{-3pt}\raisebox{-15pt}[\headheight][0pt]{\vbox{\hbox to \textwidth
{\thepage \hfil \evenhead}\vskip4pt \hrule}}}
\renewcommand{\@oddhead}{
\hspace*{-3pt}\raisebox{-15pt}[\headheight][0pt]{\vbox{\hbox to \textwidth
{\oddhead \hfil \thepage}\vskip4pt\hrule}}}
\renewcommand{\@evenfoot}{}
\renewcommand{\@oddfoot}{}

\long\def\@makecaption#1#2{%
  \vskip\abovecaptionskip
  \sbox\@tempboxa{\small \textbf{#1.}\ \ #2}%
  \ifdim \wd\@tempboxa >\hsize
    {\small \textbf{#1.}\ \ #2}\par
  \else
    \global \@minipagefalse
    \hb@xt@\hsize{\hfil\box\@tempboxa\hfil}%
  \fi
  \vskip\belowcaptionskip}

\newcommand{\JNMPnumberwithin}[3][\arabic]{%
  \@ifundefined{c@#2}{\@nocounterr{#2}}{%
    \@ifundefined{c@#3}{\@nocnterr{#3}}{%
      \@addtoreset{#2}{#3}%
      \@xp\xdef\csname the#2\endcsname{%
        \@xp\@nx\csname the#3\endcsname .\@nx#1{#2}}}}%
}

\renewenvironment{proof}[1][\proofname]{\par
  \normalfont
  \topsep6\p@\@plus6\p@ \trivlist
  \item[\hskip\labelsep\textbf{%
    #1\@addpunct{.}}]\ignorespaces
}{%
  \qed\endtrivlist
}

\newcommand{\resetfootnoterule} {
  \renewcommand\footnoterule{%
  \kern-3\p@
  \hrule\@width.4\columnwidth
  \kern2.6\p@}
}

\renewcommand{\footnoterule}{}

\makeatother

\theoremstyle{definition}

\newtheorem{theorem}{Theorem}
\newtheorem{lemma}{Lemma}

\newtheorem{proposition}{Proposition}

\theoremstyle{remark}

\begin{document}

\renewcommand{\evenhead}{ {\LARGE\textcolor{blue!10!black!40!green}{{\sf \ \ \ ]ocnmp[}}}\strut\hfill 
Yan Lingjuan, Hu Xingbiao, Liu Yajie and Zhang Yingnan
}
\renewcommand{\oddhead}{ {\LARGE\textcolor{blue!10!black!40!green}{{\sf ]ocnmp[}}}\ \ \ \ \  
On the Satsuma--Mimura Equation
}

\thispagestyle{empty}
\newcommand{\FistPageHead}[3]{
\begin{flushleft}
\raisebox{8mm}[0pt][0pt]
{\footnotesize \sf
\parbox{150mm}{{\textcolor{blue!10!black!40!green}{{\bf Open Communications in Nonlinear Mathematical Physics}}}
\ \ {Special Issue: Hietarinta}, 2026\\[0.1cm]
\strut\hfill 
ocnmp:18568
pp #2\hfill {\sc #3}}}\vspace{-13mm}
\end{flushleft}}

\FistPageHead{1}{\pageref{firstpage}--\pageref{lastpage}}{ \ \ }

\strut\hfill

\strut\hfill

\copyrightnote{The authors. Distributed under a Creative Commons Attribution 4.0 International License}

\begin{center}

{\bf {\large A Special OCNMP Issue in Honour of Jarmo Hietarinta}}\\[0.2cm]
{\bf {\large on the Occasion of his 80th Birthday}}
\end{center}

\smallskip

\Name{Pole Dynamics, Linearization, and Perturbations of the Satsuma--Mimura Equation}

\Author{Yan Lingjuan\textsuperscript{1,2}, Hu Xingbiao\textsuperscript{2,3}, Liu Yajie\textsuperscript{3},
Zhang Yingnan\textsuperscript{4,*}}
\Address{%
\textsuperscript{1} Department of Mathematics and Physics, North China Electric Power University, Baoding, Hebei, China\\
\textsuperscript{2} LSEC, ICMSEC, Academy of Mathematics and Systems Science, Chinese Academy of Sciences, Beijing, China\\
\textsuperscript{3} School of Mathematical Sciences, University of Chinese Academy of Sciences, Beijing, China\\
\textsuperscript{4} Key Laboratory of NSLSCS, Ministry of Education, School of Mathematical Sciences, Nanjing Normal University, Nanjing, Jiangsu, China\\[2mm]
* Corresponding author: \normalfont\texttt{ynzhang@njnu.edu.cn}}

\Date{Received June 21, 2026; Accepted July 7, 2026}

\setcounter{equation}{0}

\smallskip

\noindent
{\bf Citation format for this Article:}\newline
Yan Lingjuan, Hu Xingbiao, Liu Yajie and Zhang Yingnan,
Pole Dynamics, Linearization, and Perturbations of the Satsuma--Mimura Equation,
{\it Open Commun. Nonlinear Math. Phys.}, Special Issue:\,Hietarinta, ocnmp:18568, \pageref{firstpage}--\pageref{lastpage}, 2026.

\strut\hfill

\noindent
{\bf The permanent Digital Object Identifier (DOI) for this Article:}\newline
{\it 10.46298/ocnmp.18568}
\strut\hfill

\begin{abstract}

\noindent 
This paper investigates the pole dynamics and perturbation theory of algebraic soliton solutions associated with the Satsuma--Mimura (SM) equation.
First, we give a qualitative analysis of the pole system associated with algebraic soliton solutions, thereby completing a point left open in \cite{Yan_2025}. We then examine three perturbations of the SM equation. The $u_x$ perturbation preserves exact linearizability and leads to an explicit shifted algebraic soliton solution. The $u_{xx}$ perturbation can be reduced to the unperturbed SM equation by a scaling transformation, which yields the corresponding pole asymptotics and soliton profile. For the genuinely nontrivial $Hu_{xx}$ perturbation, we derive the first-order perturbation equations and present a short-time numerical simulation based on an explicit Euler discretization. These results clarify how algebraic solitons of the SM equation respond to different perturbative mechanisms.

\end{abstract}

\label{firstpage}

\section{Introduction}
Integrable systems involving the Hilbert transform have attracted sustained attention. Representative examples include the Benjamin--Ono equation \cite{Benjamin_1966, Davis_Acrivos_1967, Ono_1975}, Matsuno's model equation for deep-water waves and its Pfaffian solutions \cite{Matsuno_1988, Matsuno_1989, Matsuno_1990_Backlund}, the focusing and defocusing nonlocal Schr\"odinger equations \cite{pelinovsky1995intermediate, pelinovsky1995spectral, Matsuno_2000_multiperiodic, Matsuno_2023_multiphase}, 
the quantum BO equation \cite{abanov2005quantum}, the bidirectional BO equation \cite{Abanov_2009}, the Calogero--Moser derivative Schr\"odinger equation \cite{Gerard_Lenzmann_2024}, the half-wave equation \cite{Gerard_Lenzmann_2025}, and the cubic Szeg\"o equation \cite{Gerard_2024_Szego, Matsuno_2024_cubic}. Among them is a class of C-integrable systems in which the Hilbert transform is compatible with a linearizing transformation. In \cite{Yan_2025}, we investigated solution properties of the Satsuma--Mimura (SM) equation \cite{satsuma1981exact,satsuma1985} and the Mikhailov--Novikov (MN) equation \cite{Mikhailov2003classification}. The SM equation reads
\begin{equation}
    u_t = \delta (u_{xx} - (uHu)_x), \label{eq_MNSM_SME}
\end{equation}
while the MN equation is
\begin{equation}
    H u_t = \delta (u_{xx} - (uHu)_x). \label{eq_MNSM_MNE}
\end{equation}
Under the transformation
\begin{equation}
    u = i\ln \left(\frac{f^*}{f}\right)_x,
\end{equation}
if $f$ is a polynomial whose zeros lie in the upper half-plane, 
the MN equation is linearized as
\begin{equation}
    if_t + \delta f_{xx} = 0,
\end{equation}
whereas the SM equation is linearized as
\begin{equation}
    f_t = \delta f_{xx}. \label{eq_MNSM_HE}
\end{equation}
Starting from these linear equations, and taking $f(x,0)$ to be a polynomial of degree $N$ with all zeros in the upper half-plane, one obtains algebraic soliton solutions of the MN and SM equations, and their pole dynamics can then be deduced.

The resulting pole systems have a distinctive asymptotic structure: as $t\to +\infty$, depending on the sign of the parameter $\delta$, the poles approach four distinct straight lines. In particular, for the SM equation with $\delta<0$, the poles stay in the upper half-plane for all $t\in[0,+\infty)$, which corresponds to an algebraic soliton solution that exists globally and whose amplitudes tend to a common height. This naturally raises the question of how these algebraic solitons respond to perturbations. In this paper, we study several perturbations of the SM equation and complete the qualitative analysis of the pole dynamics left open in \cite{Yan_2025}.

\section{Properties of the pole systems for the Satsuma--Mimura and Mikhailov--Novikov equations}
If $f$ is a polynomial of degree $N$, we write it in the factorized form
\begin{equation}
    f(x,t) = \frac{1}{N!}\prod_{j=1}^N \bigl(x-x_j(t)\bigr).
\end{equation}
At $t=0$, we assume $f(x,0)$ has the form
\begin{equation}
f(x,0)
=
\frac{c_N}{N!}x^N
+
\frac{c_{N-1}}{(N-1)!}x^{N-1}
+\cdots+c_0.
\end{equation}
Comparing the coefficients of the two representations yields
\begin{equation}
c_N=1,
\qquad
c_{N-1}
=
-\frac{1}{N}\sum_{j=1}^N x_j(0),
\qquad
\ldots,
\qquad
c_0
=
\frac{(-1)^N}{N!}
\prod_{j=1}^N x_j(0).
\end{equation}

In \cite{Yan_2025}, we derived the general polynomial solution of the linearized SM equation:
\begin{equation}
    f_N(x,t) = \sum_{j=0}^N \Bigl( c_j \sum_{m+2n = j} \frac{x^m(\delta t)^n}{m!n!} \Bigr). \label{eq_gensol_SME}
\end{equation}
We also obtained the asymptotic expression for the zeros of $f$ when $\delta = -\frac{1}{2}$ as $t \to +\infty$, namely
\begin{equation}
    x_j(t) \sim \alpha_j \sqrt{t} - c_{N-1},
\end{equation}
where $\alpha_j$ are the zeros of the $N$-th degree probabilists' Hermite polynomial
\begin{equation}
H^{\mathrm{Prob}}_N(x) = N! \sum_{j=0}^{\lfloor\frac{N}{2}\rfloor } \frac{(-1)^j}{j! (N-2j)!}\frac{x^{N-2j}}{2^j}.
\end{equation}
The corresponding general polynomial solution of the MN equation is
\begin{equation}
    f_N(x,t) = \sum_{j=0}^N \Bigl( c_j \sum_{m+2n = j} \frac{x^m(i\delta t)^n}{m!n!} \Bigr). \label{eq_gensol_MNE}
\end{equation}

The following scaling observation gives the corresponding statement for a general value of $\delta$.
\begin{proposition}
    If $x_j(t),\ (j=1,2,\cdots,N)$ are the zeros of the polynomial
    \[\sum_{j=0}^N c_j \sum_{m + 2n = j} \frac{x^m (-\frac{1}{2}t)^n}{m!n!}, \]
    then $\sqrt{-2\delta}\,x_j(t),\ (j=1,2,\cdots,N)$ are the zeros of the polynomial:
    \begin{equation}
    (-2\delta)^{\frac{N}{2}}\sum_{j=0}^N \frac{c_j}{(-2\delta)^{j/2}} \sum_{m + 2n = j} \frac{x^m (\delta t)^n}{m!n!}.
    \end{equation}
    \label{proposition_zeros_delta}
\end{proposition}

Using this proposition, we obtain the asymptotic representation of the zeros for both the MN and SM equations for a general $\delta$ as $t\to +\infty$.

\begin{theorem}
    Let 
    \begin{equation}
        \bar{x}_0 = -c_{N-1} = \frac{1}{N}\sum_{j=1}^N x_j(0).
    \end{equation}
    Then $\{x_j(t)\}\ (j=1,\cdots,N)$ have the following asymptotic behavior as $t\to+\infty$.
    
    For the SM equation \eqref{eq_MNSM_SME} ($\delta<0$), the pole system $\{x_j(t)\}$ $(j=1,2,\dots,N)$ of the $N$-algebraic soliton solution $u$ constructed from the polynomial solution $f$ of \eqref{eq_MNSM_HE} satisfies
    \begin{equation}
        x_j(t) \sim \sqrt{-2\delta}\,\alpha_j \sqrt{t} + \bar{x}_0, \quad t \to +\infty,
    \end{equation}
    where $\alpha_j$ are the zeros of $H^{\mathrm{Prob}}_N(x)$.
    
    For the SM equation \eqref{eq_MNSM_SME} ($\delta>0$), $\{x_j(t)\}$ satisfies
    \begin{equation}
        x_j(t) \sim i\sqrt{2\delta}\,\alpha_j \sqrt{t} + \bar{x}_0, \quad t \to +\infty.
    \end{equation}

    For the MN equation \eqref{eq_MNSM_MNE} ($\delta<0$), $\{x_j(t)\}$ satisfies
    \begin{equation}
        x_j(t) \sim e^{i\frac{\pi}{4}}\sqrt{-2\delta}\,\alpha_j \sqrt{t} + \bar{x}_0, \quad t \to +\infty.
    \end{equation}
    
    For the MN equation \eqref{eq_MNSM_MNE} ($\delta>0$), $\{x_j(t)\}$ satisfies
    \begin{equation}
        x_j(t) \sim e^{i\frac{3\pi}{4}}\sqrt{2\delta}\,\alpha_j \sqrt{t} + \bar{x}_0, \quad t \to +\infty.
    \end{equation}
    \label{proposition_MNSM_zeros}
\end{theorem}

We next complete the analysis left open in \cite{Yan_2025} by studying the pole system directly from a qualitative perspective.
It is enough to analyze the pole system of the SM equation for $\delta<0$:
\begin{equation}
    \frac{dx_j}{dt} = -2\delta \sum_{k \neq j}\frac{1}{x_j - x_k}. \label{eq_MNSM_HE_poles_1st}
\end{equation}
The other cases can be reduced to this one by rotation:
\begin{itemize}
    \item For the SM equation with $\delta>0$, set $y_j = e^{-i\frac{\pi}{2}}x_j$; then $\frac{dy_j}{dt} = 2\delta \sum_{k \neq j}\frac{1}{y_j - y_k}$.
    \item For the pole system of the MN equation $\frac{dx_j}{dt} = -2i\delta \sum_{k \neq j}\frac{1}{x_j - x_k}$ with $\delta>0$, set $y_j = e^{-i\frac{3\pi}{4}}x_j$; then $\frac{dy_j}{dt} = 2\delta \sum_{k \neq j}\frac{1}{y_j - y_k}$.
    \item For $\frac{dx_j}{dt} = -2i\delta \sum_{k \neq j}\frac{1}{x_j - x_k}$ with $\delta<0$, set $y_j = e^{-i\frac{\pi}{4}}x_j$; then $\frac{dy_j}{dt} = -2\delta \sum_{k \neq j}\frac{1}{y_j - y_k}$.
\end{itemize}
Throughout this qualitative discussion we restrict attention to the non-coalescing case. If pole coalescence occurs, the pole system ceases to provide the appropriate description and the underlying linear equation should be analyzed directly.

\begin{lemma}
\label{lemma_MNSM_conservation}
    The sum of the poles is conserved:
    \begin{equation}
        \sum_{j=1}^N x_j(t) = \sum_{j=1}^N x_j(0).
    \end{equation}
\end{lemma}
\begin{proof}
    Summing \eqref{eq_MNSM_HE_poles_1st} over $j=1,\dots,N$ gives $\sum_{j=1}^N \frac{d x_j(t)}{dt} = 0$, which proves the assertion.
\end{proof}

\begin{lemma}
\label{lemma_MNSM_2}
    At any $t_0>0$, let $x_l$ be the pole with the smallest imaginary part and $x_s$ the pole with the largest imaginary part. Then
    \begin{align}
        \left.\frac{d \mathrm{Im}\, x_l}{dt}\right|_{t=t_0} &\geq 0,\\
        \left.\frac{d \mathrm{Im}\, x_s}{dt}\right|_{t=t_0} &\leq 0.
    \end{align}
    Thus the lowest pole cannot move downward, while the highest pole cannot move upward.
\end{lemma}
\begin{proof}
    From \eqref{eq_MNSM_HE_poles_1st},
    \begin{align}
        \left.\frac{d \mathrm{Im}\, x_l}{dt}\right|_{t = t_0} &=\left. -2\delta \sum_{j\neq l} \frac{\mathrm{Im}\, x_j - \mathrm{Im}\, x_l}{|x_l - x_j|^2}\right|_{t = t_0}\geq 0,\\
        \left.\frac{d \mathrm{Im}\, x_s}{dt}\right|_{t = t_0} &=\left. -2\delta \sum_{j\neq s} \frac{\mathrm{Im}\, x_j - \mathrm{Im}\, x_s}{|x_s - x_j|^2}\right|_{t = t_0}\leq 0.
    \end{align}
    This proves the claim.
\end{proof}

\begin{lemma}
\label{lemma_MNSM_3}
    At time $t$, label the poles by increasing real part as $x_1(t), x_2(t), \dots, x_N(t)$ (equal real parts are allowed; the assignment of indices to specific poles may vary with time). Let $d_{j} = |\mathrm{Re}\, x_{j+1}(0) - \mathrm{Re}\, x_j(0)|$. Then
    \begin{equation}
        |\mathrm{Re}\,(x_{j+1}(t) - x_j(t))| \leq \sqrt{-8\delta(N-1)t + d_j^2}.
    \end{equation}
\end{lemma}
\begin{proof}
    From the equations for $x_j$ and $x_{j+1}$,
    \begin{align}
        \frac{d \mathrm{Re}\, x_{j+1}(t)}{dt} &\leq -2\delta(N-1)\frac{1}{\mathrm{Re}\,(x_{j+1}(t) - x_j(t))}, \\
        \frac{d \mathrm{Re}\, x_j(t)}{dt} &\geq 2\delta(N-1)\frac{1}{\mathrm{Re}\,(x_{j+1}(t) - x_j(t))}.
    \end{align}
    Therefore,
    \begin{equation}
        \frac{d(\mathrm{Re}\, x_{j+1}(t) - \mathrm{Re}\, x_j(t))}{dt} \leq -4\delta(N-1)\frac{1}{\mathrm{Re}\,(x_{j+1}(t) - x_j(t))}.
    \end{equation}
    Integrating with respect to $t$ gives the stated bound.
\end{proof}

\begin{theorem}
\label{theorem_MNSM_toL}
    Let $L = \frac{1}{N}\sum_{j=1}^N \mathrm{Im}\, x_j(0)$. Then
    \begin{equation}
        \lim_{t\to +\infty}\mathrm{Im}\, x_j(t) = L, \quad j=1,\dots,N.
    \end{equation}
\end{theorem}

\begin{proof}
    Define 
    \begin{equation}
        y(t) = \max\{|\mathrm{Im}\, x_j(t) - \mathrm{Im}\, x_k(t)| : j,k=1,\dots,N\} = \mathrm{Im}\, x_s(t) - \mathrm{Im}\, x_l(t),
    \end{equation}
    and 
    \begin{equation}
    d = \max\{d_j : j=1,\dots,N-1\}.
    \end{equation}

    At every differentiability point of $y(t)$, we have
    \begin{equation}
    \begin{aligned}
        \frac{d y(t)}{dt} = &\frac{d\,\mathrm{Im}\, x_s(t)}{dt} - \frac{d\,\mathrm{Im}\, x_l(t)}{dt}\\
        =&\sum_{j\neq s}^N 2\delta \frac{\mathrm{Im}\, x_s(t) - \mathrm{Im}\, x_j(t)}{|x_s(t) - x_j(t)|^2} - \sum_{j\neq l}^N 2\delta \frac{\mathrm{Im}\, x_l(t) - \mathrm{Im}\, x_j(t)}{|x_l(t) - x_j(t)|^2}\\
        \leq& 2\delta \frac{\mathrm{Im}\, x_s(t) - \mathrm{Im}\, x_l(t)}{|x_s(t) - x_l(t)|^2} - 2\delta \frac{\mathrm{Im}\, x_l(t) - \mathrm{Im}\, x_s(t)}{|x_l(t) - x_s(t)|^2}\\
        =& 4\delta \frac{\mathrm{Im}\, x_s(t) - \mathrm{Im}\, x_l(t)}{|x_s(t) - x_l(t)|^2}\\
        =& 4\delta \frac{y(t)}{y(t)^2 + (\mathrm{Re}\, (x_l(t) - x_s(t)))^2}
    \end{aligned}
    \end{equation}
    By Lemma \ref{lemma_MNSM_2}, $y(t)$ is non-increasing, so $y(t) \leq y(0)$ for all $t$. 
    By the labeling convention in Lemma \ref{lemma_MNSM_3},
    the poles are ordered by increasing real part. Hence
    $\bigl|\mathrm{Re}\, (x_l(t) - x_s(t))\bigr| < \bigl((N-1)\max\{ \mathrm{Re}\,(x_{j+1}(t) - x_j(t))\}\bigr)$.
    Consequently,
    \begin{equation}
        \frac{y(t)}{y(t)^2 + (\mathrm{Re}\, (x_l(t) - x_s(t)))^2} \geq \frac{y(t)}{y(0)^2 + \bigl((N-1)\max\{ \mathrm{Re}\,(x_{j+1}(t) - x_j(t))\}\bigr)^2}.
    \end{equation}
    Since $\delta < 0$, applying Lemma \ref{lemma_MNSM_3} again gives
    \begin{equation}
    \begin{aligned}
        \frac{d y(t)}{dt} &\leq \frac{4\delta\, y(t)}{{y(0)}^2 + \bigl((N-1)\max\{ \mathrm{Re}\,(x_{j+1}(t) - x_j(t))\}\bigr)^2} \\
        &\leq \frac{4\delta\, y(t)}{{y(0)}^2 + (N-1)^2\bigl(d^2 - 8\delta(N-1)t\bigr)}.
    \end{aligned}
    \end{equation}
    Since $y(t)\ge 0$, 
    \begin{equation}
    \frac{d}{dt}\ln y(t)
    \leq
    \frac{4\delta}
    {y(0)^2+
    (N-1)^2
    \bigl(d^2-8\delta(N-1)t\bigr)}.
    \end{equation}
    The integral of the right-hand side diverges to $-\infty$ as $t\to+\infty$. Therefore
    \begin{equation}
    \lim_{t\to+\infty}y(t)=0.
    \end{equation}
    Finally, by Lemma \ref{lemma_MNSM_conservation},
    \[
        \frac{1}{N}
        \sum_{j=1}^N
        \mathrm{Im}\, x_j(t)
        =L
    \]
    for all $t\geq 0$. Since the difference between any two imaginary parts tends to zero, it follows that
    \begin{equation}
        \lim_{t\to+\infty}
        \mathrm{Im}\, x_j(t)=L,
        \qquad j=1,\dots,N.
    \end{equation}
    
\end{proof}

\begin{theorem}
\label{theorem_MNSM_toinfinity}
    For $x_1$ and $x_N$ (defined as in Lemma \ref{lemma_MNSM_3}),
    \begin{align}
        \lim_{t\to +\infty}\mathrm{Re}\, x_1(t) &= -\infty,\\
        \lim_{t\to +\infty}\mathrm{Re}\, x_N(t) &= +\infty.
    \end{align}
\end{theorem}
\begin{proof}
    Let $y(0) = \max\{|\mathrm{Im}\, x_j(0) - \mathrm{Im}\, x_k(0)|\}$ and $c = \frac{{y(0)}^2}{|\mathrm{Re}\, x_N(0) - \mathrm{Re}\, x_1(0)|}$.
    Set $x(t) = \mathrm{Re}\, x_N(t) - \mathrm{Re}\, x_1(t)$. Then
    \begin{equation}
        \frac{d x(t)}{dt} \geq -2\delta\frac{x(t)}{{y(0)}^2 + {x(t)}^2}
                       \geq -2\delta\frac{1}{c + x(t)},
    \end{equation}
    which yields $x(t) \geq \sqrt{-4\delta t + (c + x(0))^2} - c$.
    If either $x_1(t)$ or $x_N(t)$ remained bounded, Lemma \ref{lemma_MNSM_conservation} would be contradicted. Hence the limits hold.
\end{proof}

For the SM equation \eqref{eq_MNSM_SME} with $\delta<0$, Theorem \ref{theorem_MNSM_toL} implies that as $t\to +\infty$ all poles converge to a horizontal line, while Theorem \ref{theorem_MNSM_toinfinity} shows that the leftmost and rightmost poles tend to $-\infty$ and $+\infty$, respectively. Rotating back gives the pole behavior in the remaining three cases:
\begin{itemize}
    \item For the SM equation with $\delta>0$, the poles approach the vertical line 
    \[
    x = \frac{1}{N}\sum_{j=1}^N \mathrm{Re}\, x_j(0).
    \]
    \item For the MN equation with $\delta>0$, the poles approach the line 
    \[
    y = -x + \frac{1}{N}\sum_{j=1}^N (\mathrm{Re}\, x_j(0) + \mathrm{Im}\, x_j(0)).
    \]
    \item For the MN equation with $\delta<0$, the poles approach the line 
    \[
    y = x + \frac{1}{N}\sum_{j=1}^N (\mathrm{Im}\, x_j(0) - \mathrm{Re}\, x_j(0)).
    \]
\end{itemize}
In all cases, the poles on the two sides of the limiting line move to infinity along that line; hence the pole set is unbounded.

These pole dynamics are reflected directly in the corresponding algebraic soliton solutions. For the SM equation \eqref{eq_MNSM_SME} ($\delta<0$), the algebraic soliton solution $u$ constructed from the polynomial $f_N$ evolves as $t\to +\infty$ into $N$ waves of equal amplitude, moving symmetrically about the origin.

\section{Review of methods for perturbed solitons}
Before turning to particular perturbations of the SM equation, we recall two standard perturbative approaches for soliton equations.

Keener and McLaughlin \cite{Keener_1977} developed a Green-function method for soliton perturbations. For a nonlinear equation
\begin{equation}
    \partial_t u + K(u) = \epsilon f(u),\quad 0\le \epsilon \le 1,
    \label{eq_pertu_general_u}
\end{equation}
where $K$ is a nonlinear operator. One assumes
\begin{equation}
    u \simeq u_0 + \epsilon u_1 + \cdots,
\end{equation}
with $u_0$ satisfying
\begin{equation}
    \partial_t u_0 + K(u_0) = 0.
    \label{eq_r0_original}
\end{equation}
Then $u_1$ satisfies
\begin{equation}
    (L(u_0))u_1 \equiv \partial_t u_1 + (\delta K(u_0))u_1 = f(u_0),\quad u_1|_{t=0}=0,
\end{equation}
where $(\delta K(u_0))u_1 = \lim_{\epsilon \to 0}\frac{K(u_0 + \epsilon u_1) - K(u_0)}{\epsilon}$. A Green function can then be used to invert $L(u_0)$. Define the linear operator $\hat{G}(t,t')$ from a Hilbert space $\mathcal{H}$ to itself by
\begin{equation}
\begin{aligned}
    &(L(u_0))\hat{G} \equiv \partial_t \hat{G} + (\delta K(u_0))\hat{G} = 0,\quad 0<t'<t,\\
    &\lim_{t \to t'}\hat{G}(t,t') = I.
\end{aligned}
\end{equation}
Let $G(x,t|x',t')$ be the kernel of $\hat{G}(t,t')$. Then $u_1$ can be written as
\begin{equation}
    u_1(x, t) = \int_0^t \hat{G}(t, t') f(u_0(\cdot, t'))\,dt'
              = \int_0^t \bigl(G(x,t|\cdot,t'), f(u_0(\cdot, t'))\bigr)\, dt',
\end{equation}
where $(\cdot,\cdot)$ denotes the $L^2$ inner product. To keep $u_1$ bounded over time scales $O(1/\epsilon)$, i.e., for fixed $\tau$,
\begin{equation}
    \lim_{\epsilon \to 0} \epsilon u_1(\cdot, \tau/\epsilon) = 0,
    \label{eq_r1_require}
\end{equation}
one usually allows the parameters of $u_0$ to depend on $\epsilon t$. Write $u_0(x,t,\epsilon t)$; then $u_1$ satisfies
\begin{equation}
\begin{aligned}
    &(L(u_0))u_1 = F(u_0), \\
    &F(u_0) = f(u_0) - \frac{1}{\epsilon}\bigl(\partial_t u_0 + K(u_0)\bigr),
\end{aligned}
\end{equation}
and
\begin{equation}
    u_1(t) = \int_0^t \bigl(\hat{G}(t,t')F(u_0(t'))\bigr)(x)\, dt'. \label{eq_r1}
\end{equation}
The problem is therefore reduced to choosing the slow evolution of the parameters so that \eqref{eq_r1_require} holds.

In this approach, an explicit expression for $\hat{G}$ is needed. Setting $\tilde{K}(u_0) = \delta K(u_0)$ and $\tilde{K}'$ by $(f,\tilde{K}g) = -(\tilde{K}' f,g)$, for inverse-scattering-solvable equations satisfying $\tilde{K}'(u_0)=\tilde{K}(u_0)$ \cite{Keener_1977, Keener_1977_Green}, one splits $\hat{G}(t,t') = \hat{G}_d(t,t') + \hat{G}_c(t,t')$, where $\hat{G}_d$ projects onto the discrete subspace spanned by the soliton parameters. For an $N$-soliton solution $u_0$ with $2N$ free parameters $\{p_j\}$, the set
\begin{equation}
    \Bigl\{ \frac{\partial u_0}{\partial p_j},\ j=1,\dots,2N \Bigr\}
    \label{eq_pertu_discrete_space}
\end{equation}
spans this discrete subspace. The kernel of $\hat{G}_d(t,t')$ then has the form
\begin{equation}
    G_d(x, t|x', t') = \sum_{j=1}^{2N} A_j(x,t) \frac{\partial u_0}{\partial p_j}(x', t').
\end{equation}
For large $t'$, the dependence of $u_0$ on $x'$ and $t'$ is through combinations $x'-c_j t'$, so for large $t$, $u_1$ behaves roughly as $\iint \mathcal{F}(x'-c_j t')\,dx'\,dt'$. To avoid secular growth, one requires $\int \mathcal{F}(x'-c_j t')\,dx' =0$, i.e., $F(u_0)$ to be orthogonal to the discrete subspace \eqref{eq_pertu_discrete_space}. This gives
\begin{equation}
    \Bigl( \frac{\partial u_0}{\partial p_k}, f(u_0) - \sum_{j=1}^{2N} \frac{\partial p_j}{\partial \tau} \frac{\partial u_0}{\partial p_j} \Bigr) = 0, \quad k=1,\dots,2N,
\end{equation}
with $\tau = \epsilon t$.

For cases in which $\tilde{K}'(u_0) \neq \tilde{K}(u_0)$, Tanaka \cite{Tanaka_1980, Matsuno_1995_perturbation} introduced a more direct method. Continuing to consider \eqref{eq_pertu_general_u}, assume
\begin{equation}
    u(x,t) = u_0(x,t,t_1) + \epsilon u_1(x,t,t_1) + \cdots, \quad t_1 = \epsilon t,
\end{equation}
with $u_0$ satisfying $\partial_t u_0 + K(u_0)=0$. Define $L(u_0) = \partial_t + \delta K(u_0)$; then
\begin{equation}
    (L(u_0))u_1 = f(u_0) - \frac{\partial u_0}{\partial t_1}.
\end{equation}
Let $L^* = \partial_t + (\delta K(u_0))'$ be the adjoint operator. If a function $g$ satisfies
\begin{equation}
    L^*(g) = 0, \label{eq_pertu_L*g}
\end{equation}
then combining with the equation for $u_1$ gives
\begin{equation}
    \frac{\partial}{\partial t}\Bigl( \int_{-\infty}^\infty g u_1\,dx \Bigr) = \int_{-\infty}^\infty g \cdot \Bigl( f(u_0) - \frac{\partial u_0}{\partial t_1} \Bigr)\,dx.
\end{equation}
When $u_0$ is a soliton solution and $g$ depends on $u_0$, the right-hand side becomes time-independent for large $t$. If it were nonzero, $\int g u_1 dx$ would grow linearly in $t$, which would be inconsistent with a regular perturbation expansion. Hence one imposes
\begin{equation}
    \int_{-\infty}^\infty g \cdot \Bigl( f(u_0) - \frac{\partial u_0}{\partial t_1} \Bigr)\,dx = 0.
\end{equation}
For an $N$-soliton solution, if $2N$ independent solutions $g$ of \eqref{eq_pertu_L*g} can be found, they determine the slow evolution of the $2N$ parameters of $u_0$ with respect to $t_1$.

\section{Perturbations of the Satsuma--Mimura equation}
In the approaches recalled above, the crucial step is to obtain $L^*(u_0)$ and solve
\begin{equation}
    (L^*(u_0))g = 0.
\end{equation}
For the SM equation \eqref{eq_MNSM_SME}, a direct computation gives
\begin{equation}
    L^*(u_0) = \partial_t - \delta\bigl( -\partial_x^2 - (Hu_0)\partial_x + H(u_0\partial_x) \bigr),
\end{equation}
so the equation for $g$ is
\begin{equation}
    g_t = \delta\bigl( -g_{xx} - g_x Hu_0 + H(u_0 g_x) \bigr). \label{eq_pertu_g}
\end{equation}
Solving \eqref{eq_pertu_g} directly is difficult. In the following three subsections, we instead combine the linearized equation with direct numerical computation to examine perturbations of the form $u_x$, $u_{xx}$, and $Hu_{xx}$.

\subsection{The $u_x$ perturbation}
Consider the perturbed equation
\begin{equation}
    u_t = \delta(u_{xx} - (uHu)_x) + \epsilon u_x.
\end{equation}
Using the transformation $u = i\bigl(\ln \frac{f^*}{f}\bigr)_x$, we obtain the linearized equation
\begin{equation}
    f_t = \delta f_{xx} + \epsilon f_x.
\end{equation}
Let $f = \prod_{j=1}^N \bigl(x - x_j(t) - \epsilon x_{1,j}(t)\bigr)$. The pole system becomes
\begin{equation}
    \frac{d x_j}{dt} + \epsilon \frac{d x_{1,j}}{dt} = -2\delta \sum_{k\neq j} \frac{1}{x_j - x_k + \epsilon(x_{1,j} - x_{1,k})} - \epsilon.
\end{equation}
Since the $x_j$ satisfy the unperturbed pole system
\begin{equation}
    \frac{d x_j}{dt} = -2\delta \sum_{k\neq j} \frac{1}{x_j - x_k},
\end{equation}
and $x_{1,j}(0)=0$, we obtain $x_{1,j}(t) = -t$, hence $f = \prod_{j=1}^N \bigl( x - x_j(t) + \epsilon t \bigr)$.

Equivalently, the change of variables $t=t',\ x = x' - \epsilon t'$ transforms $f_t - \epsilon f_x = \delta f_{xx}$ into $f_{t'} = \delta f_{x'x'}$. The solution of the latter equation is $\prod_{j=1}^N (x' - x_j(t'))$, so the solution of the former is $\prod_{j=1}^N (x + \epsilon t - x_j(t))$.

The corresponding algebraic soliton solution is
\begin{equation}
    u(x,t) = \sum_{j=1}^N \frac{2\,\mathrm{Im}\, x_j(t)}{(x - \mathrm{Re}\, x_j(t) + \epsilon t)^2 + (\mathrm{Im}\, x_j(t))^2},
\end{equation}
where $x_j(t)$ are the poles of the unperturbed system.

\subsection{The $u_{xx}$ perturbation}
We next consider
\begin{equation}
    u_t = (\delta + \epsilon)u_{xx} - \delta (uHu)_x. \label{eq_pertu_eq2}
\end{equation}
Set $x = k x'$ with $k = \frac{\delta+\epsilon}{\delta}$; the equation becomes
\begin{equation}
    u_t = \frac{\delta^2}{\delta+\epsilon}\bigl( u_{x'x'} - (uHu)_{x'} \bigr). \label{eq_pertu_eq2_trans}
\end{equation}
By Proposition \ref{proposition_zeros_delta}, the zeros $x'_j(t)$ of \eqref{eq_pertu_eq2_trans} have the asymptotic form
\begin{equation}
    x'_j(t) \sim \sqrt{-2\frac{\delta^2}{\delta+\epsilon}}\,\alpha_j \sqrt{t} + i\,\frac{\sum_{j=1}^N \mathrm{Im}\, x'_j(0)}{N}.
\end{equation}
Therefore the poles of \eqref{eq_pertu_eq2} satisfy
\begin{equation}
\begin{split}
    x_j(t) &\sim \frac{\delta+\epsilon}{\delta}\Bigl( \sqrt{-2\frac{\delta^2}{\delta+\epsilon}}\,\alpha_j \sqrt{t} + i\,\frac{\sum_{j=1}^N \mathrm{Im}\, x'_j(0)}{N} \Bigr)\\
    &\sim \sqrt{\frac{\delta+\epsilon}{\delta}}\,\sqrt{-2\delta}\,\alpha_j \sqrt{t} + i\,\frac{\sum_{j=1}^N \mathrm{Im}\, x_j(0)}{N}\\
    &\sim \Bigl(1 + \frac{\epsilon}{2\delta}\Bigr)\sqrt{-2\delta}\,\alpha_j \sqrt{t} + i\,\frac{\sum_{j=1}^N \mathrm{Im}\, x_j(0)}{N}.
\end{split}
\label{eq_pertu2_xj}
\end{equation}
For large $t$, the algebraic soliton solution behaves like
\begin{equation}
    u(x,t) \sim \sum_{j=1}^N \frac{2\,\mathrm{Im}\, x_j(t)}{\bigl(x - (1+\frac{\epsilon}{2\delta})\mathrm{Re}\, x_j(t)\bigr)^2 + (\mathrm{Im}\, x_j(t))^2}.
\label{eq_pertu2_u}
\end{equation}

A more precise description, avoiding reliance on the asymptotic formula alone, is as follows. Consider the linearized SM equation. Suppose its initial condition is $f(x,0) = \prod_{j=1}^N (x - x_j(0))$. If we modify this initial condition to
\begin{equation}
f(x,0) = \prod_{j=1}^N \left(x - \sqrt{\frac{\delta}{\delta + \epsilon}}\, x_j(0) \right),
\end{equation}
then let $\tilde{x}_j(t)$ ($j=1,\dots,N$) be the poles of the SM equation obtained from this modified initial condition. The poles $x'_j(t)$ of equation \eqref{eq_pertu_eq2_trans} are then given by
\begin{equation}
x'_j(t) = \sqrt{\frac{\delta}{\delta + \epsilon}}\,\tilde{x}_j(t),
\end{equation}
and consequently the poles $x_j(t)$ of the perturbed equation \eqref{eq_pertu_eq2} satisfy
\begin{equation}
x_j(t) = \frac{\delta + \epsilon}{\delta} x'_j(t) = \sqrt{\frac{\delta + \epsilon}{\delta}}\,\tilde{x}_j(t).
\end{equation}
As $t\to +\infty$, these poles reproduce the asymptotic expressions \eqref{eq_pertu2_xj} and \eqref{eq_pertu2_u}.

\subsection{The $Hu_{xx}$ perturbation}
Finally, consider the perturbed equation
\begin{equation}
    u_t = \delta(u_{xx} - (uHu)_x) + \epsilon Hu_{xx}. \label{eq_pertu_upertu}
\end{equation}
This perturbation does not appear to preserve the above linearization.
Instead of solving \eqref{eq_pertu_g}, we use the formal expansion
\begin{equation}
    u(x,t) = u_0(x,t) + \epsilon u_1(x,t) + \epsilon^2 u_2(x,t) + \cdots, \label{eq_pertu_assume}
\end{equation}
and substitute it into \eqref{eq_pertu_upertu}. At first order in $\epsilon$ this gives
\begin{equation}
\begin{aligned}
    {u_1}_t &= \delta\bigl( {u_1}_{xx} - (u_0 H u_1)_x - (u_1 H u_0)_x \bigr) + H{u_0}_{xx},\\
    H{u_1}_t &= \delta\bigl( H{u_1}_{xx} + (u_0 u_1)_x - (Hu_0 Hu_1)_x \bigr) - {u_0}_{xx}.
\end{aligned}
\end{equation}
Since the algebraic soliton solution of the SM equation develops singularities in finite time when $\delta>0$, we restrict the simulation of $u_1$ to $\delta<0$. Even in this case, the term $\delta {u_1}_{xx}$ may cause numerical instability, so the following computation is restricted to short times. We use explicit Euler time stepping and centered finite differences in space:
\begin{equation}
    \begin{aligned}
        \frac{{u_1}^{k+1}_j - {u_1}^k_j}{\Delta t} =& \delta \Bigl( \frac{{u_1}^k_{j+1} - 2 {u_1}^k_j + {u_1}^k_{j-1}}{{\Delta x}^2} - \frac{{u_0}^{k}_{j+1}{Hu_1}^k_{j+1} - {u_0}^k_{j-1}{Hu_1}^k_{j-1}}{2\Delta x} \\
        &- \frac{{u_1}^k_{j+1}{Hu_0}^k_{j+1} - {u_1}^k_{j-1}{Hu_0}^k_{j-1}}{2\Delta x} \Bigr) + \frac{{Hu_0}^k_{j+1} - 2{Hu_0}^k_j + {Hu_0}^k_{j-1}}{{\Delta x}^2},\\[4pt]
        \frac{{Hu_1}^{k+1}_j - {Hu_1}^k_j}{\Delta t} =& \delta \Bigl( \frac{{Hu_1}^k_{j+1} - 2 {Hu_1}^k_j + {Hu_1}^k_{j-1}}{{\Delta x}^2} + \frac{{u_0}^{k}_{j+1}{u_1}^k_{j+1} - {u_0}^k_{j-1}{u_1}^k_{j-1}}{2\Delta x} \\
        &- \frac{{Hu_1}^k_{j+1}{Hu_0}^k_{j+1} - {Hu_1}^k_{j-1}{Hu_0}^k_{j-1}}{2\Delta x} \Bigr) - \frac{{u_0}^k_{j+1} - 2{u_0}^k_j + {u_0}^k_{j-1}}{{\Delta x}^2}.
    \end{aligned}
\end{equation}

We take $\delta = -\frac{1}{2}$, $\Delta t = 10^{-6}$, $\Delta x = 0.2$ and compute up to $t = 0.2$. The profile of $u_1$ is shown in Figure~\ref{fig_pertu}, together with $u_0$ for comparison.

The term $\delta u_{xx}$ with $\delta<0$ causes numerical instability.
Constructing a numerical scheme that solves the SM equation stably over long time intervals is a nontrivial task. A structure-preserving scheme will be discussed in a separate work.

\begin{figure}
    \centering
    \includegraphics[width = 0.55\linewidth]{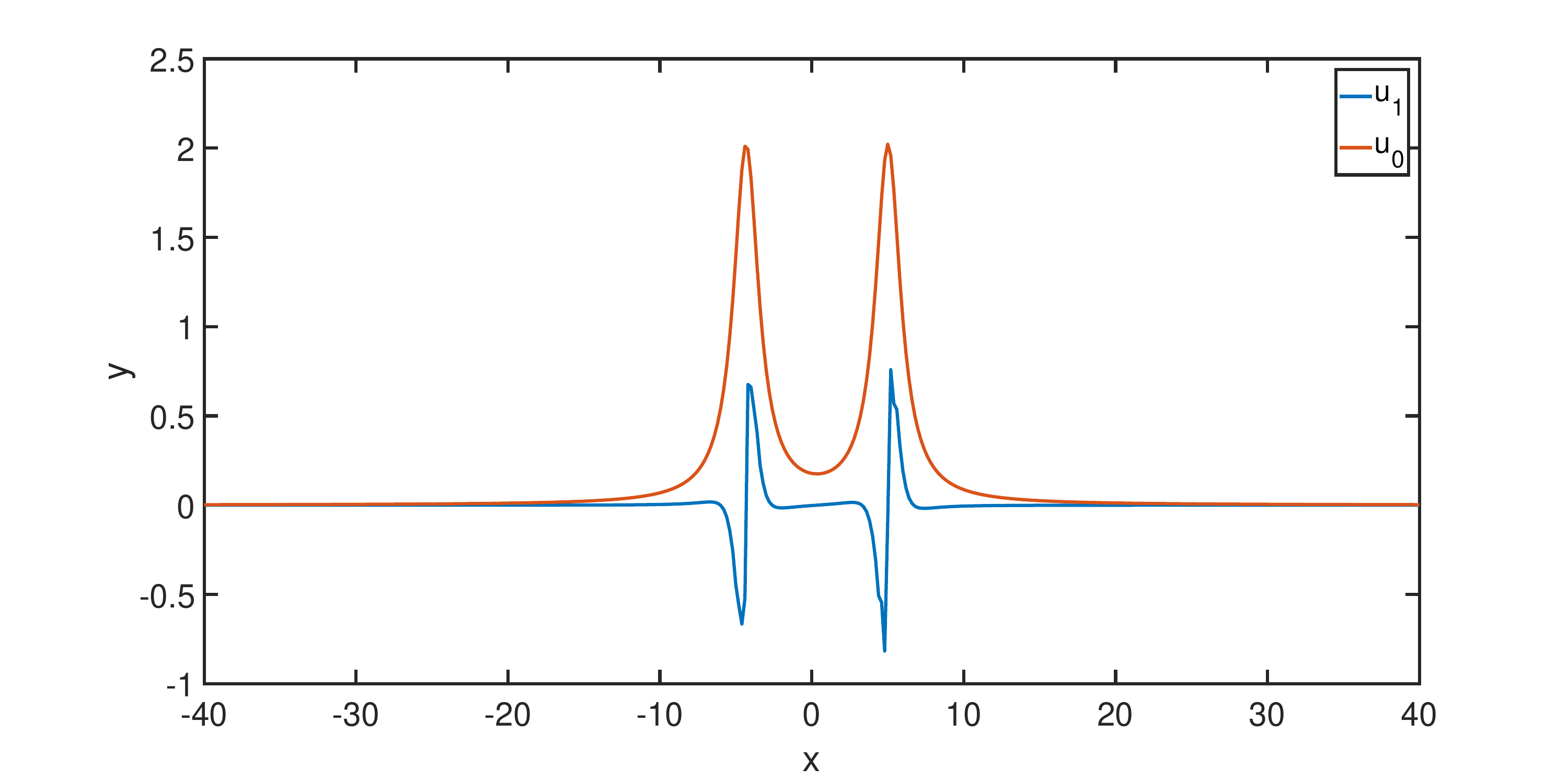}
    \caption{Short-time numerical simulation of $u_0$ and the first-order correction $u_1$ at $t=0.2$ for equation \eqref{eq_pertu_upertu} under the assumption \eqref{eq_pertu_assume}.}
    \label{fig_pertu}
\end{figure}

\section{Conclusion}
We have analyzed the pole dynamics of the Satsuma--Mimura equation through the associated pole system and then used this structure to study three perturbations of the SM equation. For the $u_x$ perturbation, exact linearizability is preserved by a simple shift transformation, which gives an explicit algebraic soliton solution. For the $u_{xx}$ perturbation, a scaling transformation reduces the perturbed equation to the unperturbed SM equation and yields the asymptotic pole locations and soliton profile. For the non-linearizable $Hu_{xx}$ perturbation, we derived the first-order perturbation equations and carried out a short-time simulation using an explicit Euler scheme. These results show how algebraic solitons respond to different perturbative mechanisms and provide a basis for further analytical and numerical work.



\label{lastpage}
\end{document}